\documentclass[a4paper]{llncs}

\usepackage{enumerate}

\usepackage{amssymb}
\usepackage{amsmath}
\usepackage[dvipdf]{graphicx}
\usepackage{url}
\usepackage[retainorgcmds]{IEEEtrantools}
\usepackage[usenames,dvipsnames,svgnames,table]{xcolor}
\usepackage{graphicx}
\usepackage{psfrag}
\usepackage{amsmath}

\newcommand{\signed}%
    {{\unskip\nobreak\hfill\penalty50
      \hskip2em\hbox{}\nobreak\hfil $\blacksquare$
      \parfillskip=0pt \finalhyphendemerits=0 \par}}

\newcommand{\shorten}[1]{}

\begin{document}
\pagestyle{empty}
\title{Balanced XOR-ed Coding }
\author{
Katina Kralevska, Danilo Gligoroski and
Harald {\O}verby}
%
%
%
%
%
%
%
%
\institute{
Department of Telematics, Faculty of Information Technology, Mathematics and Electrical Engineering, Norwegian University of Science and Technology, Trondheim, Norway,
\\{katinak@item.ntnu.no, danilog@item.ntnu.no, haraldov@item.ntnu.no} 
}

\maketitle \pagestyle{empty}

\vspace{-0.5cm}
\begin{abstract}  This paper concerns with the construction of codes over $GF(2)$ which reach the max-flow for single source multicast acyclic networks with delay. The coding is always a bitwise XOR of packets with equal lengths, and is based on highly symmetrical and balanced designs. For certain setups and parameters, our approach offers additional plausible security properties: an adversary needs to eavesdrop at least max-flow links in order to decode at least one original packet.
\end{abstract}
\vspace{-0.7cm}
\

\section{Introduction}

Encoding and decoding over $GF(2)$ is more energy efficient than encoding and decoding in any other larger field. Recent studies concerning several new techniques in network coding \cite{Ahlswede2000} (Linear Network Coding (LNC) \cite{Li2003,journals/ton/KoetterM03} and Random Linear Network Coding (RLNC) \cite{Ho2006}) confirmed that encoding and decoding over $GF(2)$ are up to two orders of magnitude less energy demanding and up to one order of magnitude faster than the encoding/decoding operations in larger fields \cite{Hassan2009,pedersen2008,pedersen}.
\shorten{
In the last decade, network coding has attracted a lot of interest among researchers due to its significant advantages such as improved throughput, robustness, scalability, and security. These performance gains are achieved by combining algebraically the packets at the nodes in the network, either at sources or at intermediate nodes. For network coding to be effective, the condition of receiving packets from multiple paths is obligatory.

Network coding was first introduced in the seminal paper by Ahlswede et al. \cite{Ahlswede2000}, followed by the remarkable works of Koetter and M\'{e}dard in \cite{journals/ton/KoetterM03}, Li et al. \cite{Li2003}, and Ho et al., \cite{Ho2006}. Both linear network coding (LNC) \cite{Li2003} and random linear network coding (RLNC) \cite{Ho2006} achieve the capacity when the field size is sufficiently large. However, performing operations in larger finite fields is costly and complex. In \cite{Hassan2009}, \cite{pedersen2008}, and \cite{pedersen}, it has been reported that RLNC can be several orders of magnitude more energy demanding and up to one order of magnitude slower than the encoding done by simple XOR operations.}

The high computational complexity of packet encoding and decoding over large finite fields and its high energy cost which makes it unsuitable for practical implementation are the main motivation to seek for coding techniques only with XOR operations. The first theoretical work was done by Riis in \cite{Riis} who showed that every solvable multicast network has a linear solution over $GF(2)$. Afterwards, XOR coding in wireless networks was presented in \cite{Katti2008}, where the main rule is that a node can XOR $n$ packets together only if the next hop has all $n-1$ packets. A more general network coding problem which is called index coding is considered in \cite{Sprintsonmatroid,indexcoding}. In \cite{indexcoding} the authors address the coding problem by proposing coding over $GF(2)$. The encoding scheme is based on bitwise XORing by adding redundant bits, and the decoding scheme is based on a simple but bit after bit sequential back substitution method.

The main contribution of our work is a construction of codes over $GF(2)$ by using combinatorial designs (Latin squares and Latin rectangles) \cite{Stinson1999}. Its lower
computation and energy cost makes it suitable for practical
implementation on devices with limited processing and energy
capacity like mobile phones
and wireless sensors. We will illustrate the construction of codes by the following simple example.
\begin{figure}[!ht]
\vspace{-0.5cm}
\centering
\includegraphics[width=2.5in]{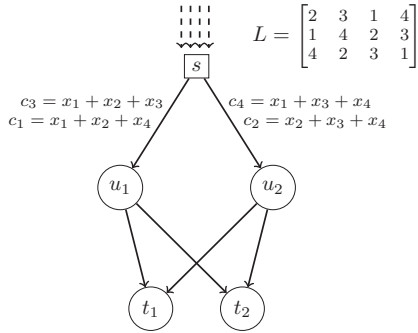}
\caption{An example of balanced XOR coding where the source sends combinations of source packets (combined as the column of the Latin rectangle). The intermediate nodes just forward the data to the sink nodes.}
\label{voved}
\end{figure}

\vspace{-0.6cm}

\begin{example}
We use the following strategy (Fig. \ref{voved}): the source $s$ performs bitwise XOR of packets with equal length based on the incidence matrix of a Latin rectangle $L$. Each column of $L$ represents a combination of source packets $x_i$, $i=1,\ldots,4$, in a coded packet $c_i$, $i=1,\ldots,4$. In the first phase, the packets $c_1$ and $c_2$ are sent, and in the second phase, the packets $c_3$ and $c_4$  are sent. The intermediate nodes $u_1$ and $u_2$ forward the coded packets to the sink nodes $t_1$ and $t_2$ which decode the packets by using the inverse matrix of the incidence matrix of $L$. The sink nodes need only to know the combination of source packets in each received packet. Note that the max-flow in the network is achieved.
\end{example}
\vspace{-0.3cm}
Routinely as in other coding approaches, this information is included in the header of each coded packet.  Since in this paper we use diversity coding performed just by the source nodes, there is no need for updating the coefficients in the header at each intermediate node. The length of the prepended header vector is negligible compared to the length of the packet.

The construction of our codes was not motivated by security issues, therefore the security is not the main goal in this paper. However, it turns out that for certain setups and parameters, our approach offers additional plausible security properties. The plausible security properties that accompany our approach are not based on hard mathematical problems in modern cryptology (for example factoring of large integers or discrete logarithm problems or on the Shamir's secret sharing algorithm). We show that if an eavesdropper wants to reconstruct at least one original packet, then the number of eavesdropped links should be equal to the max-flow of the network. Bhattad and al. \cite{Bhattad} make similar observations when network coding is implemented so that a weekly secure network coding is achieved.

The rest of the paper is organized as follows: Section \ref{SecPrelim} presents the notation and the mathematical background that are used in the following sections. The construction of codes is presented in Section \ref{construction}. Section \ref{security} illustrates the security features of our approach. Sections \ref{conclusion} concludes the paper.

\section{Notation and Mathematical Background} \label{SecPrelim}
We define a communication network as a tuple $N=(V, E, S, T)$ that consists of:
\begin{itemize}
	\item a finite directed acyclic multigraph $G=(V, E)$ where $V$ is the set of vertices and $E$ is the set of edges,
	\item a set $S \subset V$ of sources,
	\item a set $T \subset V$ of sink nodes.
\end{itemize}
Assume that vertex $s\in S$ sends $n$ source packets to vertex $t\in T$ over disjoint paths. A minimal cut separating $s$ and $t$ is a cut of the smallest cardinality denoted as $\text{mincut}(s, t)$. The packets are sent in several time slots, i.e., phases denoted as $p$. The maximum number of packets that can be sent in a phase from $s$ to $t$ is denoted as $\text{maxflow}(t)$. The Max-Flow Min-Cut Theorem \cite{MaxFlowMinCut} indicates that $\text{mincut}(s, t) = \text{maxflow}(t)$. The multicast capacity, i.e., the maximum rate at which $s$ can transfer information to the sink nodes, cannot exceed the capacity of any cut separating $s$ from the sink nodes. A network is solvable when the sink nodes are able to deduce the original packets with decoding operations. If the network is solvable with linear operations we say that the network is linearly solvable.

\subsection{XOR-ed coding}

First we recall that in \cite{Riis}, Riis showed that every solvable multicast network has a linear solution over $GF(2)$ in some vector dimension. The essence of his proof relies on the fact that any two finite fields with the same cardinality are isomorphic. Thus, instead of working in a finite field $GF(2^n)$ for which the conditions of the linear-code multicast (LCM) theorem \cite[Th. 5.1]{Li2003} are met, he showed that it is possible to work in the isomorphic vector space $GF(2)^n$ that is an extension field over the prime field $GF(2)$. We formalize the work in the vector space $GF(2)^n$ with the following:

\begin{definition}
A XOR-ed coding is a coding that is realized exclusively by bitwise XOR operations between packets with equal length. Hence, it is a parallel bitwise linear transformation of $n$ source bits $x=(x_1,\ldots,x_n)$ by a $n \times n$ nonsingular matrix $K$, i.e., $y = K \cdot x$.
\end{definition}

In \cite{Riis} it was also shown that there are simple network topologies where encoding in $GF(2)$ cannot reach the network capacity with the original bandwidth or by sending data in just one phase. However, it was shown that the network capacity by XOR-ed coding can be achieved either by increasing the bandwidth or the number of phases so that they match the dimension of the extended vector space $GF(2)^n$. In this paper we take the approach to send data in several phases $p$ instead of increasing the bandwidth.


\begin{theorem}\label{XORed-unicast}
For any linearly solvable network topology with $\text{maxflow}(t_1)>1$, the sufficient condition for a single sink $t_1$ to reach its capacity in each of $p$ phases by XOR-ed coding is to receive $n$ linearly independent packets $x=(x_1, \ldots, x_n)$, where $n=p \times \text{maxflow}(t_1)$.
\end{theorem}
\begin{proof}
Assume that the network topology is linearly solvable. That means there exists a vector space $GF(2)^n$ where we can encode every $n$ source bits with a bijective function $K$, i.e., $y=K \cdot x$. Having in mind that the source $s$ succeeds to send $n$ encoded packets to $t_1$ in $p$ phases, and the max-flow in the network is $\text{maxflow}(t_1)>1$, we have that $n=p \times \text{maxflow}(t_1)$ and the sink $t_1$ receives $n$ packets after $p$ phases via $\text{maxflow}(t_1)$ disjoint paths. In order to have a successful recovery of the initial $n$ packets, the received packets should be linearly independent.
\end{proof}
\vspace{-0.3cm}
Based on Theorem \ref{XORed-unicast} we can prove the following:
\begin{theorem}\label{XORed-duocast}
For any linearly solvable network topology and for any two sinks $T=\{t_1, t_2\}$ that have $\text{maxflow}(t)=\text{maxflow}(t_1)=\text{maxflow}(t_2)$, there always exists a XOR-ed coding for $n=p \times \text{maxflow}(t)$ packets that achieves the multicast capacity in each of $p$ phases.
\end{theorem}
\vspace{-0.254cm}
\begin{proof}
For the sink $t_1$ we apply Theorem \ref{XORed-unicast} and find one XOR-ed coding that achieves the capacity in each of $p$ phases. Let us denote by $U_1=\{ u_{1,i} |\ \text{there is an} \\ \text{edge}\ (u_{1,i}, t_1) \in E \}$ the nodes that are directly connected and send packets to the sink node $t_1$. We have that $|U_1|=\text{maxflow}(t)$, and the set of $n$ packets is partitioned in $\text{maxflow}(t)$ disjoint subsets $Y_{1,1}, \ldots, Y_{1,\text{maxflow}(t)}$ each of them having $p$ packets. The subset $Y_{1,i}$ comes from the node $u_i$, $i=1,\ldots, \text{maxflow}(t)$.

The set $U_2=\{ u_{2,i} | \text{there is an edge} (u_{2,i}, t_2) \in E \}$ is a set of nodes that are directly connected and send packets to the sink node $t_2$. We denote the intersection between the sets of nodes $U_1$ and $U_2$ as $U_{1,2}=U_1 \bigcap U_2$.
The following three situations are considered:
\begin{enumerate}
  \item There are no mutual nodes that send packets to both sinks $t_1$ and $t_2$, i.e., $U_{1,2}=\emptyset$. In that case find one partition of the set of $n$ packets in $\text{maxflow}(t)$ disjoint subsets $\Gamma_1=\{Y_{2,1} \ldots, Y_{2,\text{maxflow}(t)} \}$ each of them having $p$ packets. The sets of packets $Y_{2,j}$ are delivered to the sink $t_2$ via the node $u_{2, j}$, $j=1,\ldots,\text{maxflow}(t)$. The multicast capacity for the sink $t_2$ is achieved in each of $p$ phases.
  \item There are nodes that send packets to both sinks $t_1$ and $t_2$, i.e., $U_{1,2}=\{ u_{(1,2)_{\nu_1}}, \ldots, u_{(1,2)_{\nu_k}} \}$. Denote the nodes that are in $U_2 \backslash U_1 = \{ u_{2_{\nu_1}}, \ldots\\ \ldots ,u_{2_{\nu_{\text{maxflow}(t)-k}}} \}$. In that case, the sink $t_2$ receives from the nodes in $U_{1,2}$ the same packets that are delivered to the sink $t_1$. The number of the remaining packets that have to be delivered to $t_2$ is exactly $p \times (\text{maxflow}(t) - k)$. Find one partition of $\text{maxflow}(t) - k$ disjoint subsets $\Gamma_2=\{Y_{2,1} \ldots, Y_{2,\text{maxflow}(t) - k} \}$ each of them having $p$ packets. The sets of packets $Y_{2,j}$ are delivered to the sink $t_2$ via the node $u_{2,\nu_j}$, $j=1,\ldots,\text{maxflow}(t)-k$. The multicast capacity for the sink $t_2$ is achieved in each of $p$ phases.
  \item All the nodes that send packets to the sink $t_1$, send packets to the sink $t_2$ as well, i.e., $U_{1,2}=U_1 \bigcap U_2 = U_1$. In that case, the sink $t_2$ receives from the nodes in $U_{1,2}$ the same packets that are delivered to the sink $t_1$. The multicast capacity for the sink $t_2$ is achieved in each of $p$ phases.
\end{enumerate}
\end{proof} Note that the proof of Theorem \ref{XORed-duocast} is similar to the work by Jaggi et al. \cite{Jaggi} where they discuss a construction of general codes using simple algorithms.

As a consequence of Theorems \ref{XORed-unicast} and \ref{XORed-duocast} we can post the following:
\begin{theorem}\label{XORed-multicast}
For any linearly solvable network topology and for any set of $N$ sinks $T=\{t_1, \ldots, t_N\}$ that have $\text{maxflow}(t)=\text{maxflow}(t_1)=\ldots=\text{maxflow}(t_N)$, there always exists a XOR-ed coding for $n=p \times \text{maxflow}(t)$ packets that achieves the multicast capacity in each of $p$ phases.
\end{theorem}
\begin{proof}
(\textit{Sketch}) First, we recall the construction of generic linear codes presented in the LCM theorem in \cite[Th. 5.1]{Li2003}. Second, we use the transformation to equivalent codes over $GF(2)^n$ as it was shown in \cite{Riis}. Then, the proof is a straightforward application of the mathematical induction by the number of sinks $N$. Let us suppose that the claim of the theorem is correct for $N-1$ sinks. By adding a new $N$-th sink we consider again three possible situations as in Theorem \ref{XORed-duocast}.
\end{proof}


\section{Construction of XOR-ed Coding} \label{construction}

In this section we describe the construction of codes over $GF(2)$. Instead of working with completely random binary matrices, in the remaining part of this paper we work with nonsingular binary matrices that have some specific structure related to randomly generated Latin square or Latin rectangle. We do not reduce the space of possible random linear network encoding schemes, since the number of Latin squares and Latin rectangles of order $n$ increases proportionally with factorial of $n$. Therefore, in our approach we have virtually an endless repository of encoding schemes that have the benefits from both worlds: they are randomly generated, but they have a certain structure and offer plausible security properties.

In order to introduce our approach, we briefly use several definitions that the reader can find in \cite{Stinsonbook} and \cite{Combinatorial}.

\begin{definition}
A Latin square of order $n$ with entries from an $n$-set X is an $n\times n$
array $L$ in which every cell contains an element of X such that every row of $L$ is a
permutation of X and every column of $L$ is a permutation of X.
\label{Latinsquare}
\end{definition}

\begin{definition}
A $k\times n$ Latin rectangle is a $k\times n$ array (where $k \leq n$) in which each cell contains
a single symbol from an $n$-set X, such that each symbol occurs exactly once in each
row and at most once in each column.
\label{Latinrectangle}
\end{definition}

\begin{figure}[!ht]
\vspace{-0.5cm}
\centering
\includegraphics[width=2.5in]{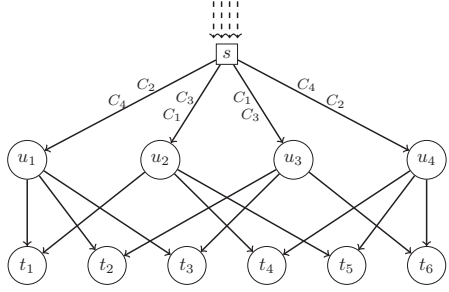}

\caption{A 4-dimensional binary linear multicast in a single source multicast network with delay}
\label{vtora}
\end{figure}

For generating a Latin square, one can always start with a permutation of $n$ elements that is a trivial $1\times n$ Latin rectangle and can use the old Hall's marriage theorem \cite{H123} to construct new rows until the whole Latin square is completed. However, this approach does not guarantee that the generated Latin squares are chosen uniformly at random. In order to generate Latin squares of order $n$ that are chosen uniformly at random we use the algorithm of Jacobsen and Matthews \cite{JCD:JCD3}. Further, in our approach we sometimes split the Latin square into two Latin rectangles (upper and lower), and work with the algebraic objects (matrices or block designs) that are related to either the upper or the lower Latin rectangle.

As a convention, throughout this paper, the number of packets $n$ that are sent from the source is equal to the number of columns in the Latin square or Latin rectangle.


\begin{example}
\label{Example01}
As shown in Fig. \ref{vtora}, we assume that the source wants to send four packets $x_1,\ldots,x_4$ to the sink nodes, and that each sink node has maxflow($t_k$) = 2, ($k=1,\ldots,6$). The sink nodes receive data from different pair of intermediate nodes, $u_i$, ($i=1,\ldots,4$). Our aim is all six sink nodes to be able to reconstruct the source packets that are exclusively coded in $GF(2)$.

Let us take the following Latin square and split it into two Latin rectangles:
$$\small
L =
\begin{bmatrix}
2 & 4 & 1 & 3\\
1 & 3 & 2 & 4\\
3 & 2 & 4 & 1\\
\hline
4 & 1 & 3 & 2\\
\end{bmatrix}.
$$
Each column from the $3\times 4$ upper Latin rectangle represents a combination of source packets in a coded packet $c_i$, $i=1,\ldots,4$.
Using the incidence matrix $M$ of the Latin rectangle the source computes the coded packets.
\end{example}

\begin{definition}
Let $(X,A)$ be a design where $X = \{x_1,\ldots, x_v\}$ and $A = \{A_1,\ldots,A_b\}$. The incidence matrix of $(X,A)$ is the $v\times b$ 0-−1 matrix $M = (m_{i,j})$ defined by the rule
$
m_{i,j} =
\begin{cases}
1, & \text{if}\  x_i\in A_j, \nonumber\\
0, & \text{if}\  x_i\notin A_j. \nonumber\\
\end{cases}
$
\label{incidencematrix}
\end{definition}

\begin{proposition}
The incidence matrix $M = (m_{i,j})$ of any Latin rectangle with dimensions $k\times n$ is balanced matrix with $k$ ones in each row and each column.
\label{balanced}
\end{proposition}
\begin{proof}
From the definition of the incidence matrix it follows that the number of ones in each row is equal to the number of elements $k$ in each column of the Latin rectangle. On the other hand, since each row of the Latin rectangle is a permutation of $n$ elements, and there are no elements that occur twice in each column, the number of ones in each column can be neither less nor larger than $k$.
\end{proof}

\begin{note}
The incidence matrix $M$ of a $k\times n$ Latin rectangle is always balanced. However, the inverse matrix of the incidence matrix $M^{-1}$ is not always balanced.
\label{note1}
\end{note}
\vspace{-0.75cm}
\begin{proposition}
The necessary condition an incidence matrix $M = (m_{i,j})$ of a $k\times n$ Latin rectangle to be nonsingular in $GF(2)$ is $k$ to be odd, i.e., $k=2 l + 1$.
\label{nonsingular}
\end{proposition}
\begin{proof}
Assume that $k$ is even, i.e., $k=2 l$. Recall that a matrix $M$ is nonsingular in $GF(2)$ if and only if its determinant is 1 (or it is singular if and only if its determinant is 0). Recall further the Leibniz formula for the determinant of an $n\times n$ matrix $M$: $det(M) = \sum_{\sigma \in S_n} sgn(\sigma) \prod_{i=1}^n m_{i,\sigma_i},$ where the sum is computed over all elements of the symmetric group of $n$ elements $S_n$, i.e., over all permutations $\sigma \in S_n$, and $sgn(\sigma)$ is the signature (or the parity of the permutation) whose value is $+1$ or $-1$. The elements $m_{i,\sigma_i}$ are the elements $m_{i,j}$ of the matrix $M$ where the value for the index $j=\sigma_i$ is determined as the $i$--th element of the permutation $\sigma$.

If $k=2 l$ is even, from Proposition \ref{balanced} and from the fact that operations are performed in $GF(2)$, it follows that every summand in the Leibniz formula gives an even number of nonzero products, thus the final sum must be even, i.e., the determinant in $GF(2)$ is 0.
\end{proof}

The corresponding $4\times 4$ incidence matrix of the Latin rectangle in Example \ref{Example01} is nonsingular in $GF(2)$ (Proposition~\ref{nonsingular}). $M$ is represented as
$$\small
M =
\begin{bmatrix}
1 & 1 & 1 & 0\\
0 & 1 & 1 & 1\\
1 & 1 & 0 & 1\\
1 & 0 & 1 & 1\\
\end{bmatrix}.
$$

A direct consequence from Theorem~\ref{XORed-unicast} is the following:

\begin{corollary}\label{corollaryPhases}
A sink node $t\in T$ with $\text{maxflow}(t)$ can receive $n$ source packets, encoded with the incidence matrix of a $k\times n$ Latin rectangle in $GF(2)$, in $p=\lceil {\frac{n}{\text{maxflow}(t)}}\rceil$ phases. In each phase the sink node reaches its $\text{maxflow}(t)$.
\end{corollary}

Following Corollary~\ref{corollaryPhases} the number of phases in which packets are sent depends from the total number of packets and maxflow($t_k$).

Using $M$ the source computes the vector of coded packets as
$$
\mathbf{c} = M \mathbf{x} = [c_1, c_2, c_3, c_4]^\top
$$
where $\mathbf{x} = [x_1, x_2, x_3, x_4]^\top$ is a vector of the source packets. The coded packets are XOR-ed combinations of the source packets, i.e.,
\begin{IEEEeqnarray*}{rCl}\scriptsize
\begin{aligned}
& c_1 = x_1\oplus x_2\oplus x_3,\\
& c_2 = x_2\oplus x_3\oplus x_4,\\
& c_3 = x_1\oplus x_2\oplus x_4,\\
& c_4 = x_1\oplus x_3\oplus x_4.\\
\end{aligned}
\end{IEEEeqnarray*}
The source further prepends the information from the incidence matrix to each of the coded packets. The vector of packets that are sent becomes as follows: $\mathbf{C} = \{(1, 2, 3, c_1), (2, 3, 4, c_2), (1, 2, 4, c_3), (1, 3, 4, c_4)\}$ $=\{C_1, C_2, C_3, C_4\}$.
The sink nodes receive in each phase a pair of different packets as shown in Table~\ref{tabela}. Their buffer should be large enough to store the received packets $C_i$, $i=1,\ldots,4$.
\begin{table}[t]

\vspace{-0.5cm}
\caption{Description of receiving coded packets in each phase at the sink nodes}
\vspace{-0.5cm}
\begin{center}
\begin{tabular}{|c||@{\ }c@{\ }|@{\ }c@{\ }|@{\ }c@{\ }|@{\ }c@{\ }|@{\ }c@{\ }|@{\ }c@{\ }|}
\hline
&$t_1$&$t_2$&$t_3$&$t_4$&$t_5$&$t_6$\\
\hline
First phase&$C_4, C_1$&$C_4, C_3$&$C_4, C_3$&$C_1, C_2$&$C_1, C_2$&$C_3, C_2$\\
\hline
Second phase&$C_2, C_3$&$C_2, C_1$&$C_2, C_1$&$C_3, C_4$&$C_3, C_4$&$C_1, C_4$\\
\hline
\end{tabular}
\end{center}
\vspace{-0.3cm}
\label{tabela}
\end{table}
The decoding at the sink nodes is performed by $M^{-1}$. Each sink node computes $M^{-1}$ from the prepended indexes.
The original packets $x_i$, $i=1,\ldots,4$, are reconstructed as $ \mathbf{x} = M^{-1} \mathbf{c}.$ Note that although our approach is similar to \cite{Riis}, we use a systematic selection of the encoding functions and we do not send plain packets on the disjoint paths.



\section{Additional Plausible Security Properties of the Balanced XOR-ed Coding} \label{security}

The work with incidence matrices related to randomly generated Latin rectangles is actually a work with balanced block designs. However, as we noted in Note~\ref{note1}, it is not necessary both the incidence matrix and its inverse matrix to be completely balanced. If we are interested in the complexity of decoding and the security issues when an adversary can successfully decode some sniffed packets, then the easiest way to address these issues is to give equal level of security to all encoded packets. In our approach this can be easily achieved by switching the roles of the incidence matrix and its inverse matrix: the encoding of the source packets is done with the inverse matrix of the incidence matrix and decoding of the coded packets is done with the incidence matrix. By applying this approach, decoding of any of the source packets requires an equal number of coded packets.

\begin{corollary}
For each value of $\text{maxflow}(t)$ and a number of source packets $n$ which is multiple of $\text{maxflow}(t)$, there exists a Latin rectangle with $n-1$ or $n-2$ rows and its incidence matrix can be used for decoding.
\label{corollary1}
\end{corollary}

Due to Proposition~\ref{nonsingular}, when $n$ is even the necessary requirement for a nonsingular incidence matrix is the Latin rectangle to have $n-1$ rows. When $n$ is odd the necessary requirement for a nonsingular incidence matrix is the Latin rectangle to have $n-2$ rows.

\begin{theorem}\label{TheoremListening1}
 When decoding is performed with the incidence matrix from Corollary~\ref{corollary1}, any eavesdropper needs to listen at least $\text{maxflow}(t)$ links in order to decode at least one source packet.
\end{theorem}
\vspace{-0.5cm}
\begin{proof}
Assume that an adversary eavesdrops $\text{maxflow}(t) - 1$ links. Since the incidence matrix used for decoding is related to a Latin rectangle with $n-1$ or $n-2$ rows, eavesdropping ``just'' $\text{maxflow}(t) - 1$ links is not sufficient for the adversary to receive at least one subset of $n-1$ or $n-2$ packets from which he/she can decode at least one original packet.
\end{proof}

Another remark that can be given about our approach is that the number of XOR operations between different packets (both in the source and in the sink nodes) is relatively high. We can address that remark by using Latin rectangles with smaller number of rows as a trade-off between the number of encoding/decoding operations and the ability of an adversary to decode a source packet. Namely, the encoding and decoding efforts at the source and sink node are the highest when encoding and decoding requires $n-1$ or $n-2$ packets. In order to decrease the number of operations at the nodes, the Latin rectangle should have $k\leq n-2$ rows. However, we are interested to reduce the number $k$ without reducing the number of links that have to be listened by an eavesdropper in order to decode at least one original packet. The following theorem gives the necessary and sufficient condition for that to happen:

\begin{theorem}\label{TheoremOptimizedSecurity}
Let the coding be done by $M^{-1}$ obtained from a Latin rectangle $L_{k\times n}$ of size $k \times n$, where $k\leq n-2$. Further, assume that the transfer is done by sending $n$ packets from $s$ to $t$ in $p=\lceil {\frac{n}{\text{maxflow}(t)}}\rceil$ phases on $\text{maxflow}(t)$ disjoint paths and let the sets of indexes of the packets sent via $i$-th disjoint path are denoted by $S_i, i=1,\ldots,\text{maxflow}(t)$. A necessary and sufficient condition for an eavesdropper to need to listen at least $\text{maxflow}(t)$ links in order to decode at least one original packet is:
\begin{equation}\label{ConditionForOptimizedSecurity}
\small
\forall j \in \{1,\ldots,n\}, \forall i \in \{1,\ldots,\text{maxflow}(t)\}: L_{k,j} \cap S_i \neq \emptyset,
\end{equation}
where $L_{k,j}, j \in \{1,\ldots,n\}$ is the set of elements in the $j$-th column of the Latin rectangle $L_{k\times n}$.
\end{theorem}
\begin{proof}
To show that the condition (\ref{ConditionForOptimizedSecurity}) is necessary assume that an eavesdropper needs $\text{maxflow}(t) - 1$ links in order to decode one original packet $x_l$, and let us denote by $S_m$ the set of indexes of the packets sent via the disjoint path that was not listened by the eavesdropper. This means that for the $l$-th column $L_{k,l}$ of the Latin rectangle $L_{k\times n}$: $L_{k,l} \cap S_m = \emptyset$ which violates the condition (\ref{ConditionForOptimizedSecurity}).

To show that the condition (\ref{ConditionForOptimizedSecurity}) is sufficient, let us denote by $S_{i,j}=L_{k,j} \cap S_i$, $j \in \{1,\ldots,n\}, i \in \{1,\ldots,\text{maxflow}(t)\}$. It is sufficient to notice that $S_{i,j}$ are disjunctive partitions for every set $L_{k,j}$, i.e., $$\forall j \in \{1,\ldots,n\}: \bigcup_{i=1}^{\text{maxflow}(t)} S_{i,j}=L_{k,j}$$ and $$\forall j_1, j_2 \in \{1,\ldots,n\}: S_{i,j_1} \cap S_{i,j_2} = \emptyset.$$ Since $|L_{k,j}|=k$, and the encoding of original $n$ packets is done by $M^{-1}$, it follows that an eavesdropper can decode any original packet only by listening at least $\text{maxflow}(t)$ links.
\end{proof}

\newpage
\begin{example}
\label{Example02}
We present an example that illustrates the security in our approach. The goal is to achieve secrecy\footnote{We use here the term \emph{secrecy} as it is used in \cite[Ch.7 pp. 185]{Medardbook} }
so that a passive adversary is able to reconstruct $n$ source packets only when at least $\text{maxflow}(t)$ links are eavesdropped. By sending XOR-ed packets on disjoint paths (exploiting the path diversity), an adversary is unable to decode the message although several paths are eavesdropped. Let us consider the network shown in Fig.\ref{treta}, where a source $s$ communicates with two sinks $t_1$ and $t_2$ with the help of intermediate nodes $u_i$, $i=1, 2, 3$, and sends twelve packets to $t_1$ and $t_2$. Packets are sent in four phases since maxflow($t$) = 3. Let us use the following $5\times 12$ Latin rectangle:
$$\small
L_{5\times 12}=
\left[
\begin{array}{@{\ }c@{\ }c@{\ }c@{\ }c@{\ }c@{\ }c@{\ }c@{\ }c@{\ }c@{\ }c@{\ }c@{\ }c@{\ }c@{\ }}
\color{green}{4} & \color{red}{2} & \color{green}{11} & \color{green}{8} & \color{blue}{12} & \color{blue}{1} & \color{green}{9} & \color{blue}{5} & \color{red}{10} & \color{red}{7} & \color{blue}{6} & \color{red}{3} \\
\color{red}{2}   & \color{green}{8} & \color{blue}{12} & \color{red}{3} & \color{blue}{6} & \color{red}{10} & \color{green}{4} & \color{green}{11} & \color{blue}{5} & \color{blue}{1} & \color{green}{9} & \color{red}{7} \\
\color{red}{3}   & \color{green}{4} & \color{red}{2} & \color{green}{9} & \color{green}{11} & \color{blue}{12} & \color{blue}{5} & \color{blue}{6} & \color{red}{7} & \color{green}{8} & \color{red}{10} & \color{blue}{1} \\
\color{green}{9} & \color{red}{10} & \color{blue}{1} & \color{blue}{6} & \color{red}{3} & \color{red}{7} & \color{red}{2} & \color{green}{8} & \color{green}{4} & \color{green}{11} & \color{blue}{12} & \color{blue}{5} \\
\color{blue}{6}  & \color{blue}{5} & \color{red}{7} & \color{red}{10} & \color{blue}{1} & \color{green}{11} & \color{green}{8} & \color{red}{3} & \color{blue}{12} & \color{green}{4} & \color{red}{2} & \color{green}{9} \\
\end{array} \right].
$$


The colors of indexes in $L_{5\times 12}$ correspond to the colors of the packets as they are sent in Fig \ref{treta}. If the sink nodes reconstruct the source packets with $M^{-1}$, then not all packets have the same level of decoding complexity. That is demonstrated with relations (\ref{ExampeEncoding12}) and (\ref{ExampeDecoding12}). For instance, to decode $x_4, x_7$ and $x_8$ nine coded packets are needed, while to decode $x_5$ and $x_{11}$ just three packets are needed. The goal is to avoid this non-balanced complexity in the decoding. Therefore, as in Theorem \ref{TheoremOptimizedSecurity} the encoding is done by $M^{-1}$ and the decoding by $M$. When $s$ computes the vector of coded packets as $\mathbf{c} = M^{-1} \mathbf{x},$ then the coded packets $c_i$, $i=1,\ldots,12$ are XOR-ed combinations of different number of source packets. Consequently, decoding of packets is done with a balanced matrix, i.e., $\mathbf{x} = M \mathbf{c}.$
\begin{figure}
\centering
\includegraphics[width=2.5in]{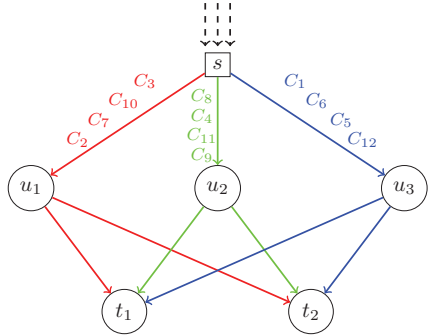}

\caption{Routing of 12 packets for secure coding when decoding is performed with 5 coded packets}
\label{treta}
\end{figure}

\begin{minipage}{0.35\textwidth}
\begin{IEEEeqnarray}{l}\label{ExampeEncoding12}\scriptsize
\begin{aligned}
c_1    = & x_2\oplus x_3\oplus x_4\oplus x_6 \oplus x_9,\\
c_2    = & x_2\oplus x_4\oplus x_5\oplus x_8 \oplus x_{10},\\
c_3    = & x_1\oplus x_2\oplus x_7\oplus x_{11} \oplus x_{12},\\
c_4    = & x_3\oplus x_6\oplus x_8\oplus x_9 \oplus x_{10},\\
c_5    = & x_1\oplus x_3\oplus x_6\oplus x_{11} \oplus x_{12},\\
c_6    = & x_1\oplus x_7\oplus x_{10}\oplus x_{11} \oplus x_{12},\\
c_7    = & x_2\oplus x_4\oplus x_5\oplus x_8 \oplus x_9,\\
c_8    = & x_3\oplus x_5\oplus x_6\oplus x_8 \oplus x_{11},\\
c_9    = & x_4\oplus x_5\oplus x_7\oplus x_{10} \oplus x_{12},\\
c_{10} = & x_1\oplus x_4\oplus x_7\oplus x_8 \oplus x_{11},\\
c_{11} = & x_2\oplus x_6\oplus x_9\oplus x_{10} \oplus x_{12},\\
c_{12} = & x_1\oplus x_3\oplus x_5\oplus x_7 \oplus x_9.\\
\end{aligned}\hspace{0.75cm}
\end{IEEEeqnarray}
\vspace{0.2cm}
\end{minipage}
\begin{minipage}{0.35\textwidth}
\begin{IEEEeqnarray}{l}\label{ExampeDecoding12}\scriptsize
\hspace{0.5cm}\begin{aligned}
x_1    = & c_2\oplus c_5\oplus c_{10}\oplus c_{11} \oplus c_{12},\\
x_2    = & c_1\oplus c_3\oplus c_5\oplus c_6 \oplus c_7\oplus c_9 \oplus c_{10},\\
x_3    = & c_3\oplus c_5\oplus c_6\oplus c_7 \oplus c_8\oplus c_9 \oplus c_{12},\\
x_4    = & c_4\oplus c_5\oplus c_6 \oplus c_7\oplus c_8\oplus c_9 \oplus c_{10}\oplus c_{11} \oplus c_{12},\\
x_5    = & c_1\oplus c_2\oplus c_4,\\
x_6    = & c_1\oplus c_2\oplus c_4\oplus c_6 \oplus c_7\oplus c_{10} \oplus c_{11},\\
x_7    = & c_2\oplus c_3\oplus c_4\oplus c_5 \oplus c_6\oplus c_7 \oplus c_8 \oplus c_{11}\oplus c_{12},\\
x_8    = & c_1\oplus c_3\oplus c_5\oplus c_7 \oplus c_8\oplus c_9 \oplus c_{10} \oplus c_{11} \oplus c_{12},\\
x_9    = & c_1\oplus c_2\oplus c_5\oplus c_9 \oplus c_{10},\\
x_{10} = & c_1\oplus c_5\oplus c_7\oplus c_9 \oplus c_{10},\\
x_{11} = & c_1\oplus c_7\oplus c_8,\\
x_{12} = & c_3\oplus c_4\oplus c_5\oplus c_7 \oplus c_9.\\
\end{aligned}\hspace{0.5cm}
\end{IEEEeqnarray}
\vspace{0.2cm}
\end{minipage}



Assume that the routing is as follows: on the first path the source sends ($\color{red}C_3, C_{10}, C_7, C_2$), on the second path ($\color{green}C_8, C_{4}, C_{11}, C_{9}$) and ($\color{blue}C_1, C_{6}, C_{5}, C_{12}$) on the third path as it is shown in Fig.\ref{treta}. We use three different colors for the packets sent to three disjoint paths in order to demonstrate the essence of the proof of Theorem~\ref{TheoremOptimizedSecurity}. Note that all colors are present in every column of the Latin rectangle $L_{5\times 12}$. This corresponds to the condition (\ref{ConditionForOptimizedSecurity}) in Theorem~\ref{TheoremOptimizedSecurity}.  In order to reconstruct at least one source packet, an adversary must eavesdrop at least 3 links.
\end{example}

\section{Conclusions}\label{conclusion}
In this paper we have presented a construction of codes over $GF(2)$ which reach the max-flow for single source multicast acyclic networks with delay. The coding is exclusively performed in $GF(2)$, i.e., it is a bitwise XOR of packets with equal lengths. The encoding and decoding are based on balanced nonsingular matrices that are obtained as incidence matrices from Latin rectangles. Balanced XOR-ed coding is of particular importance for energy and processor constraint devices. Additionally, we showed that the approach offers plausible security properties, i.e., if an eavesdropper wants to reconstruct at least one original packet, then the number of eavesdropped links must be equal to the max-flow of the network.

Possible future work includes intermediate nodes to form
coded packets, as well as building networks dynamically by
adding more and more sink nodes that reach the max-flow when the
coding is XOR-ed coding.

\section*{Acknowledgements}\label{acknowledgemets}
\vspace{-0.25cm}
We would like to thank Gergely Bicz\'{o}k for his discussions and remarks that significantly improved the paper.

\bibliographystyle{plain}
\bibliography{refer}

\end{document}